\documentclass[conference]{IEEEtran}
\IEEEoverridecommandlockouts
\makeatletter
\def\ps@headings{%
\def\@oddhead{\mbox{}\scriptsize\rightmark \hfil \thepage}%
\def\@evenhead{\scriptsize\thepage \hfil \leftmark\mbox{}}%
\def\@oddfoot{}%
\def\@evenfoot{}}
\makeatother

\usepackage{amssymb,amsmath}
\usepackage{paralist}%clock
\usepackage{graphicx}

\usepackage{color}
\usepackage{cite}

\newcommand{\F}{\mathbf{F}}

\newcommand{\C}{\mathcal{C}}
\newcommand{\N}{\mathcal{N}}

\newtheorem{theorem}{\textbf{Theorem}}
\newtheorem{lemma}[theorem]{\textbf{Lemma}}
\newtheorem{definition}[theorem]{\textbf{Definition}}

\newtheorem{remark}[theorem]{\textbf{Remark}}

\newcommand{\nix}[1]{}

\begin{document}
\title{Network Coding-based Protection  Strategies Against a Single Link Failure in Optical Networks\thanks{This
research was supported in part by grants CNS-0626741 and
CNS-0721453 from the National Science Foundation,
and a gift from Cisco Systems.}}
\author{
\authorblockN{Salah A. Aly}
\authorblockA{
Department of ECE \\Iowa State University\\ Ames, IA, 50011, USA\\Email: salah@iastate.edu} \and
\authorblockN{Ahmed E. Kamal}
\authorblockA{Department of ECE \\Iowa State University\\Ames, IA 50011, USA\\Email:
kamal@iastate.edu}
 } \maketitle

\begin{abstract}
In this paper we  develop network protection strategies  against a single link failure in optical networks. The motivation behind this work is the fact that $\%70$  of all available links in an optical network suffers from a single link failure. In the proposed protection strategies, denoted NPS-I and NPS-II, we deploy network coding and reduced capacity  on the working paths to provide a backup protection path that will carry encoded data from all sources. In addition, we provide implementation aspects and how to deploy the proposed strategies in case of an optical network with $n$ disjoint working paths.
\end{abstract}
\section{Introduction}\label{sec:intro}
One of the main services of operation networks that must be deployed
efficiently is reliability. In order to deploy a  reliable networking strategy,  the transmitted signals must be protected over unreliable
links. Link failures are common problems that might occur frequently in single and multiple operating communication circuits. In network
survivability and network resilience, one needs to design  efficient
strategies to overcome this dilemma.  Optical network survivability
techniques are classified as \emph{pre-designed protection} and
\emph{dynamic
restoration}~\cite{kamal08a,markopoulou04,somani06,zhou00}. The approach of using pre-designed protections aims to reserve enough bandwidth such that when a failure occurs, backup paths are used to reroute the transmission and  to recover the data. Examples of this approach are 1-1 and 1-N protections~\cite{kamal07a,kamal08b}. In dynamic restoration reactive strategies, capacity is not reserved. However, when the failure occurs, dynamic recovery is used to recover the data transmitted in the links that are suffered from failures. This technique does not  require preserved resources or provisioning extra paths that work in case of failure. In this work we will provide several strategies of dynamic restoration based on network coding and reduced distributed fairness capacities.

Network coding is  a powerful tool that has been recently used to
increase the throughput, capacity, security,  and performance of
communication networks. Information theoretic aspects of network coding have been investigated in~\cite{soljanin07,yeung06,ahlswede00}.  Network
coding allows the intermediate nodes not only to forward packets using
network scheduling algorithms, but also encode/decode them using algebraic primitive operations,
see~\cite{ahlswede00,fragouli06,soljanin07,yeung06},
and references therein. Network coding is used to maximize the
throughput~\cite{ahlswede00,jaggi03,koetter03}, network
capacity~\cite{ramamoorthy05,aly07e,kong07}. Also, it is robust
against packet losses and network failures\cite{ho03,lun05}, and  enhances network security and protection~\cite{jaggi07,lima06}.   It is believed that network coding will be deployed in all relay nodes and network operations.

Network protection against a single link failure (SLF) by adding one extra path has been introduced in~\cite{kamal06a,kamal07a,kamal07b}. The source
nodes are able to combine their data into a single extra path (backup
\emph{protection path}) that is used to protect all signals  on the
\emph{working paths} carrying data from all sources.  Also, protection
against multiple link failures has been presented
in~\cite{kamal08a,kamal08b,jaggi07}, where $m$ extra paths are used. In
both cases, $p$-cycles have been  used for protection against single and
multiple link failures. In this model the source nodes are assumed to send their data with a full capacity relaying on the extra paths to protect
their data. However, there are situations where the extra paths approach
might not be applicable, and one needs to design  a protection strategy
depending  solely  on the available
resources~\cite{aly08i,kamal08a,aly08preprint1,aly08patent1}.

In this work we will assume that all paths are working and adding extra
paths to the available ones is a difficult task. We apply two network
protection strategies called NPS-I and NPS-II, each of which has $(n-1)/n$ normalized network capacity. In these two strategies, we show how the
sources achieve the encoding operation and distribute their link's
capacities among them for fairness. We assume that one of the working
paths will overlap to carry encoded data, therefore,  acting as a
protection path.

In this paper we introduce a model for network protection against a single link failure  in optical networks. In this model, the network capacity will be reduced by partial factor in order to achieve the required protection. Several advantages of NPS-I and NPS-II strategies can be stated as follows.
\begin{itemize}
\item The data originated from the sources is protected without adding extra secondary paths. We assume that one of the working paths
    will act as a protection path carrying encoded data.
\item The encoding and decoding operations are achieved online with
    less computational cost at both the sources and receivers.
\item The normalized network capacity is $(n-1)/n$, which is
    near-optimal in the case of using large number of n connection
    paths.
\end{itemize}

The rest of this paper is organized as follows. In
Sections~\ref{sec:model} and~\ref{sec:terminolgoy} we present the network
model and problem setup, respectively. The definitions of the
\emph{normalized capacity}, \emph{working and protection paths} are given. In Section~\ref{sec:NPS-I-extrapaths} we present a network protection
strategy NPS-I against a single link/path failure using an extra dedicated path. In addition in Section~\ref{sec:singlefailure} we provide the
network protection strategy NPS-II which deployed network coding and
reduced capacity. The implementation aspects of NPS-I and NPS-II are
discussed in Section~\ref{sec:implementation}, and finally the paper is
concluded in Section~\ref{sec:conclusion}.

%%%%%%%%%%%%%%%%%%%%%%%%%%%%%%%%%%%%%%%%%%%%%%%%%%%%%%%%%%%%%%%%%%%%%%%%
\section{Network Model}\label{sec:model}
The network model can be describe as follows.
\begin{compactenum}[i)]
\item Let $\N$ be a network represented by an abstract graph
    $G=(\textbf{V},E)$, where $\textbf{V}$ is the set of nodes and $E$
    be   set of undirected edges. Let $S$ and $R$ is a   sets of independent sources and destinations, respectively. The set $\textbf{V}=V\cup S \cup   R$ contains the relay nodes, sources, and destinations. Assume for simplicity that $|S|=|R|=n$,
    hence the set of sources is equal to the set of receivers.

\item
A path (connection) is a set of edges connected together with a
    starting node (sender) and an ending node (receiver).  $$ L_i=
    \{(s_i,e_{1i}),(e_{1i},e_{2i}),\ldots,(e_{(m)i},r_i) \},$$ where
    $1\leq i\leq n$ and $(e_{(j-1)i},e_{ji}) \in E$ for some integer $m$.

\item
The node can be a router, switch, or an end terminal depending on the network model $\N$ and the transmission layer.

\item $L$ is a set of links $L=\{L_1,L_2,\ldots,L_n\}$ carrying the data
    from the sources to the receivers as shown in Fig.~\ref{fig:nnodes}.
    All connections have the same bandwidth, otherwise a connection with
    high bandwidth can be divided into multiple connections, each of which
    has a unit capacity. There are exactly $n$ connections.

\item
Each sender $s_i \in S$ will transmit its own data $x_i$ to a receiver $r_i$ through a connection $L_i$. Also, $s_i$ will transmit encoded data $\sum_{i}^n x_i$ to $r_i$ at different time slot if it is assigned to send the encoded data.

\item The data from all sources are sent in sessions. Each session has a
    number of time slots $n$. Hence $t_\delta^\ell$ is a value at round time slot $\ell$ in session $\delta$.

\item In this model $\N$, we consider only a single link failure, it is
    sufficient to apply the encoding and decoding operations over a finite field with two elements, we denote it $\F_2=\{0,1\}$.

\item There are at least two receivers and two senders with at least two
    disjoint paths, otherwise  the protection model can not be deployed for a single path, in which it can not protect itself.
\end{compactenum}

We will define the \emph{working} and \emph{protection} paths between two
network nodes (switches and routers) in optical networks as shown in
Fig.~\ref{fig:npaths}.

\begin{definition}
The \emph{working paths} on a network with n connection paths carry
traffic under normal operations. The \emph{Protection paths} provide an
alternate backup path to carry the traffic in case of failures. A
protection scheme ensures that data sent from the sources will reach the
receivers in case of failure incidences on the working paths.
\end{definition}

\begin{figure}[t]
 \begin{center} % Requires \usepackage{graphicx}
  \includegraphics[scale=0.8]{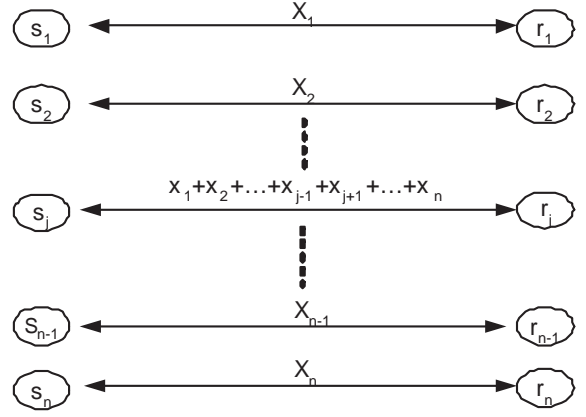}
  \caption{Network protection against a single link failure using reduced capacity and network coding. One link out of  $n$ primary links  carries encoded data.}\label{fig:nnodes}
\end{center}
\end{figure}

%%%%%%%%%%%%%%%%%%%%%%%%%%%%%%%%%%%%%%%%%%%%%%%%%%%%%%%%%%%%%%%%%%%%%55
\section{Problem Setup and Terminology}\label{sec:terminolgoy}

We assume that there is a set of $n$ disjoint connections that requires protections with $\%100$ guaranteed against a single link failure (SLF). All connections have the same bandwidth, and each link (one hop) with a
bandwidth  can be a circuit.

Every sender $s_i$ prepares a packet \emph{$packet_{s_i \rightarrow r_i}$}
 to send to the receiver $r_i$. The packet contains the sender's ID, data
$x_i^\ell$, and a round time for every session $t^\ell_\delta$ for some
integers $\delta$ and $\ell$. We have two types of packets:
\begin{compactenum}[i)]
\item Packets sent without coding, in which the sender does not need
    to perform any coding operations. For example, in case of packets
    sent without coding, the sender $s_i$ sends the following packet
    to the receiver $r_i$:
\begin{eqnarray}
packet_{s_i \rightarrow r_i}:=(ID_{s_i},x_i^\ell,t_\delta^\ell)
\end{eqnarray}

\item  Packets sent with encoded data, in which the sender needs to
    send other senders' data. In this case, the sender $s_i$ sends the following packet to the receiver $r_i$:
\begin{eqnarray}
packet_{s_i \rightarrow r_i}:=(ID_{s_i},\sum_{j=1,j\neq i}^n x_j^\ell,t^\ell_\delta).
\end{eqnarray}
The value $y_i^{\ell}=\sum_{j=1,j\neq i}^n x_j^\ell$  is computed by every sender $s_i$, in which it is able to collect the data from all other
senders and encode them using the bit-wise operation.
\end{compactenum} In either case the sender has a full capacity in the
connection link $L_i$. We will provide more elaboration in this scenario
in Section~\ref{sec:implementation}, where implementation aspects will
be discussed.

\smallbreak

We can define the network capacity from  min-cut max-flow
information theoretic view~\cite{ahlswede00}. It can be described as
follows.
\begin{definition}\label{def:capacitylink}
The unit capacity of a connecting path $L_i$ between $s_i$ and $r_i$ is
defined as: \begin{eqnarray} c_i=\left\{
      \begin{array}{ll}
        1, & \hbox{$L_i$ is \emph{active};} \\
        0, & \hbox{otherwise.}
      \end{array}
    \right.
\end{eqnarray}
The total capacity of $\N$ is given by the summation of all path
capacities. What we mean by an \emph{active}  path is that the receiver is able to receive and process signals/messages throughout this path.
\end{definition}
\smallbreak

Clearly, if all paths are active then the total capacity is $n$ and
normalized capacity is $1$. If we assume there are n disjoint paths, then, in general, the capacity of the network for the active and failed paths is computed by:
\begin{eqnarray}
C_\N=\frac{1}{n}\sum_{i=1}^n c_i.
\end{eqnarray}

\begin{figure}[t]
 \begin{center} % Requires \usepackage{graphicx}
  \includegraphics[height=5cm,width=8.6cm]{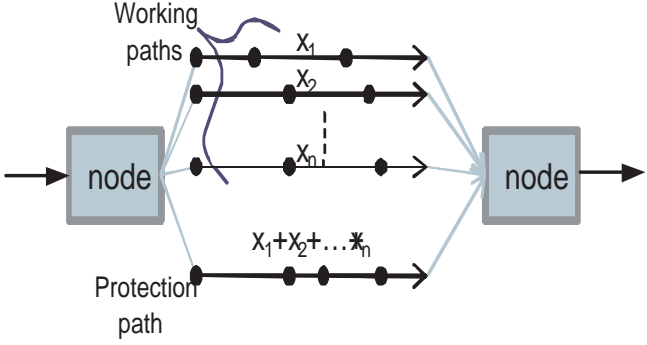} %height=4.5cm,width=8.3cm
  \caption{Network protection against a single path failure using reduced capacity and network coding. One path out of  $n$ primary paths  carries encoded data. The black points represent various other relay nodes}\label{fig:npaths}
\end{center}
\end{figure}

There have been several techniques developed to provide network
survivability. Such techniques will add additional resources for the sake
of recovery from failures. They will also depend on the time it takes  to
recover from failures, and how much delay the receiver can tolerate.
Hence, network survivability is a multi-objective problem in terms of
resource efficiency, operation cost, and agility. Optimizing these
objectives has taken much attention recently, and has led to the design of more efficient reliable networks.

%%%%%%%%%%%%%%%%%%%%%%%%%%%%%%%%%%%%%%%%%%%%%%%%%%%%%%%%%%%%55 NPCC
\section{Network Protections Against A SLF  Using Extra and Dedicated
Paths}\label{sec:NPS-I-extrapaths} Assume we have $n$ connections carrying data from a set of $n$ sources to a set of $n$ receivers. All connections
represent  disjoint paths, and the sources are independent of each other.
The author in~\cite{kamal06a,kamal07a} introduced a model for protecting an  optical network against a single link failure using an extra path
provision. The idea is to establish a new connection from the sources to
the receivers using virtual (secondary) source and virtual  (secondary) receiver. The goal of the secondary source is to collect data from all
other sources and encode it using the Xored operation.

The extra path that carries the encoded data from all sources is  one
cycle. In the encoding operations every source $s_i$ adds its value, and
the cycle starts at source $s_1$ and ends at source $s_n$. Hence, the encoded data after performing  the cycle or extra path is $X=\sum_{i=1}^n x_i$. The decoding operations are done at every receiver $r_i$ by adding the data $s_i$ received over the link $L_i$. The node $r_j$ with failed
connection $L_j$ will be able to recover the data $x_j$. Assuming all
operations are achieved over the binary finite field $\F_2$. Hence we
have: \begin{eqnarray}x_j= X-\sum_{i=1, i \neq j}^n x_i^\ell.\end{eqnarray}

\bigskip

\noindent \textbf{Protecting With Extra Paths (NPS-I):} We will describe
the network protection strategy NPS-I against a single link failure in
optical networks. Assume a source $s_i$ generates a message $x_i^\ell$ at
round time $t_\delta^\ell$. Put differently:
\begin{eqnarray}
packet_{s_i}=(ID_{s_i},x_i^\ell,t_\delta^\ell).
\end{eqnarray}
The $packet_{s_i}$ is transmitted from the source $s_i$ to a destination
$r_i$ for all $1 \leq i \leq n$. It is sent in the primary working path
$L_i$, i.e. a path that conveys an unencoded data. The secondary
protection path provisioned from a source $s$ to destination $r$ can
convey the message:

\begin{eqnarray}
packet_{s}=(ID_{s},\sum_{i=1}^n x_i^\ell,t_\delta^\ell).
\end{eqnarray}
This process is explained in Scheme~(\ref{eq:n-1protection1}).

\begin{eqnarray}\label{eq:n-1protection1}
\begin{array}{|c|cccccc|c|c|}
\hline
& \multicolumn{6}{|c|}{\mbox{ round time session 1 }}&\ldots&\ldots    \\
\hline
&1&2&3&\ldots&\ldots&n&\!\!\ldots&\ldots\\
\hline    \hline
  s_1 \rightarrow r_1 & x_1^1&x_1^2 &x_1^3&\ldots &\ldots  &x_1^n &\ldots&\ldots \\
    s_2 \rightarrow r_2 &  x_2^1&  x_2^2&x_2^3&\ldots&\ldots&x_2^{n} &\ldots&\ldots \\
s_3 \rightarrow r_3 &  x_3^1& x_3^2& x_3^3&\ldots&\ldots&x_3^{n} &\ldots&\ldots \\
 \vdots&\vdots&\vdots&\vdots&\vdots&\vdots&\vdots&\ldots&\ldots\\
    s_i \rightarrow r_i& x_i^1 & x_i^2 &\ldots&x_i^{i-1} &\ldots& x_i^{n}&\ldots&\ldots\\
 \vdots&\vdots&\vdots&\vdots&\vdots&\vdots&\vdots&\ldots&\ldots\\
   s_n \rightarrow r_n & x_n^1&x_n^2&x_n^3&\ldots&\ldots&x_{n}^n&\ldots&\ldots\\
\hline
    s \rightarrow r & y_1&y_2&y_3&\ldots&\ldots&y_{n}&\ldots&\ldots\\
\hline
\end{array}
\end{eqnarray}
All $y_j$'s are defined over $\F_2$ as: \begin{eqnarray} y_j=\sum_{i=1}^n
x_i^j.
\end{eqnarray}
 We  notice that the encoded data $y_j$ is fixed
per one session transmission and it is fixed for other sessions. This
means that the path $L_j$ is dedicated to sending all encoded data
$y_1,y_2,\ldots,y_n$.
\begin{lemma}
The normalized capacity of NPS-I of the network model $\N$ described
in Scheme~(\ref{eq:n-1protection1}) is given by
\begin{eqnarray}
\C=(n)/(n+1)
\end{eqnarray}
\end{lemma}
\begin{proof}
In every session, we  have n rounds. Furthermore, in every round there are $(n+1)$ senders with $n+1$ disjoint paths, and only one sender sends
encoded data. Therefore $\C=n^2/(n+1)n$, which gives the result.
\end{proof}

\bigskip

\noindent \textbf{Protecting Without Extra Paths:} If we do not allow an
extra path, then one of the available working paths can be used to carry
the encoded data as shown in Scheme~(\ref{eq:n-1protection2}). It  shows
that  a path $L_j$ exists which  carries the encoded data sent from the
source $s_j$ to the receiver $r_j$.

\begin{eqnarray}\label{eq:n-1protection2}
\begin{array}{|c|cccccc|c|c|}
\hline
& \multicolumn{6}{|c|}{\mbox{ round time session 1 }}&\ldots&\ldots    \\
\hline
&1&2&3&\ldots&\ldots&n&\!\!\ldots&\ldots\\
\hline    \hline
  s_1 \rightarrow r_1 & x_1^1&x_1^2 &x_1^3&\ldots &\ldots  &x_1^n &\ldots&\ldots \\
    s_2 \rightarrow r_2 &  x_2^1&  x_2^2&x_2^3&\ldots&\ldots&x_2^{n} &\ldots&\ldots \\
s_3 \rightarrow r_3 &  x_3^1& x_3^2& x_3^3&\ldots&\ldots&x_3^{n} &\ldots&\ldots \\
 \vdots&\vdots&\vdots&\vdots&\vdots&\vdots&\vdots&\ldots&\ldots\\
    s_i \rightarrow r_i& x_i^1 & x_i^2 &\ldots&x_i^{i-1} &\ldots& x_i^{n}&\ldots&\ldots\\
     \vdots&\vdots&\vdots&\vdots&\vdots&\vdots&\vdots&\ldots&\ldots\\
     s_j \rightarrow r_j & y_j^1&y_j^2&y_j^3&\ldots&\ldots&y_j^{n}&\ldots&\ldots\\
 \vdots&\vdots&\vdots&\vdots&\vdots&\vdots&\vdots&\ldots&\ldots\\
   s_n \rightarrow r_n & x_n^1&x_n^2&x_n^3&\ldots&\ldots&x_{n}^n&\ldots&\ldots\\
\hline
\hline
\end{array}
\end{eqnarray}
All $y_j^\ell$'s are defined over $\F_2$ as \begin{eqnarray}
y^\ell_j=\sum_{i=1,i\neq j}^n x_i^\ell.
\end{eqnarray}

We notice that the encoded data $y_j$ is fixed per one session
transmission but it is varied for other sessions. This means that the path $L_j$ is dedicated to send all encoded data $y_1,y_2,\ldots,y_n$.
\begin{lemma}
The normalized capacity of the network model $\N$ described
in Scheme~(\ref{eq:n-1protection2}) is given as:
\begin{eqnarray}
\C=(n-1)/n
\end{eqnarray}
\end{lemma}
The implementation aspects of this strategy is discussed in
Section~\ref{sec:implementation}.

%%%%%%%%%%%%%%%%%%%%%%%%%%%%%%%%%%%%%%%%%%%%%%%%%%%%%%%%%%%%55 NPCC
\section{Network Protection Against A SLF Using Distributed Capacity and Coding}\label{sec:singlefailure}
In this section we will provide a network protection strategy against a
single link failure using distributed fairness capacity and coding. This
strategy is called NPS-II. We will compute the network capacity in each
approach and how the optimal capacity can be written with partial delay at rounds of a given session for a sender $s_i$. In~\cite{aly08preprint1} we
will also illustrate the tradeoff between the two approaches, where there
is enough space for details.

\bigskip

\noindent \textbf{NPS-II Protecting a SLF:} We will describe the NPS-II
which protects a single link failure using network coding and reduced
capacity. Assume there is a path $L_j$ that will carry the encoded data
from the source $s_j$ to the receiver $r_j$. Consider a failed link $(u,v) \in E$,  which the path $L_i$ goes through. We would like to design an
encoding scheme such that a backup copy of data on $L_i$ can also be sent
over a different path $L_j$. This process is explained in
Scheme~(\ref{eq:n-1protectionII}), and is call it Network Protection
Strategy (NPS-II) against a single Link/path failure (SLF). The data is
sent in rounds for every session. Also, we assume the failure happens only in one path throughout a session, but different paths might suffer from
failures throughout different sessions. Indeed most of the current optical networks suffer experience a single link failure~\cite{zhou00,somani06}.

The objective of the proposed network protection strategy is to withhold
rerouting the signals or the transmitted packets due to link failures.
However, we provide strategies that utilize network coding and reduced
capacity at the source nodes. We assume that the source nodes are able to
perform encoding operations and the receiver nodes are able to perform
decoding operations. We will allow the sources to provide backup copies
that will be sent through the available paths simultaneously and in the
same existing connections.

Let $x_i^j$ be the data sent from the source $s_i$ at round time $j$ in a
session $\delta$. Also, assume $y_j=\sum_{\ell=1,\ell \neq i}^n x_\ell^j$.
Put differently
\begin{eqnarray}
y_j=x_1^j\oplus x_2^j\oplus \ldots \oplus x_n^j.
\end{eqnarray}
The protection scheme  runs in sessions as explained below. For the
$(n-1)/n$ strategy presented in Scheme~(\ref{eq:n-1protectionII}), the design issues are described as follows.
\begin{compactenum}[i)]
\item A total of $(n-1)$ link disjoint paths between $(n-1)$ senders $S$ and receivers $R$ are provisioned to carry the signals from $S$ to $R$. Each path  has the unit capacity and data unit from $s_i$ in $S$ to  $r_i$ in $R$ are sent in rounds. Data unit $x_{i}^n$ is sent from
    source $s_i$ at round (n) in a specific session.
\item A server $\mathcal{S}$ is able to collect the signals from all $n$
    sources and is able to provision $y_j=\sum_{i=1, i\neq k}^n x_{i}^j$
    at round time $j$. A single source $s_k$ is used to deliver $y_j$  to
    the receiver $r_k$. This process is achieved at one particular session.  The encoded data $y_j$ is distributed equally among all $n$
    sources.
\item In the first round time at a particular session, the data $x_i^1$ is sent from $s_i$ to $r_i$ in all paths for $i=\{1,\ldots,n\}$ and $i\neq j$. Only the source $s_j$ will send $y_j$ to the receiver $r_j$ over the path $L_j$ at round $t_\delta^\ell$.
    $$y_j=\sum_{i=1,j\neq j}^n x_i^\ell .$$
\item We always neglect the communication and computational cost between
    the senders and data collector $\mathcal{S}$, as well as the receivers    and data collector $\mathcal{R}$.
\end{compactenum}
\begin{eqnarray}\label{eq:n-1protectionII}
\begin{array}{|c|cccccc|c|c|}
\hline
& \multicolumn{6}{|c|}{\mbox{ round time session 1 }}&\ldots   \\
\hline
&1&2&3&\ldots&\ldots&n&\!\!\ldots\\
\hline    \hline
  s_1 \rightarrow r_1 & x_1^1&x_1^2 &\ldots&x_1^j &\ldots  &\!\! y_n &\ldots \\
    s_2 \rightarrow r_2 &  x_2^1&  x_2^2&\ldots&x_2^j&\ldots&\!\! x_2^{n-1} &\ldots \\
s_3 \rightarrow r_3 &  x_3^1& x_3^2&\ldots& x_3^j&\ldots&\!\! x_3^{n-1} &\ldots \\
     \vdots&\vdots&\vdots&\vdots&\vdots&\vdots&\vdots&\ldots\\
     s_j \rightarrow r_j & x_j^1&x_j^2&\ldots&y_j&\ldots&\!\!x_j^{n-1}&\ldots\\
 \vdots&\vdots&\vdots&\vdots&\vdots&\vdots&\vdots&\ldots\\
  s_{n-1} \rightarrow r_{n-1} \!\! & \!\! x_{n-1}^1&\!\! y_2&\ldots&\!\! \!\!x_{n-1}^{j-1}&\ldots&\!\!x_{n-1}^{n-1}&\ldots\\
   s_n \rightarrow r_n & y_1&x_n^1&\ldots&x_n^{j-1}&\ldots&\!\!x_{n}^{n-1}&\ldots\\
\hline
\hline
\end{array}
\end{eqnarray}
In this case $y_1=\sum_{i=1}^{n-1} x_{i}^1$ and in general  $y_j$'s are
defined over $\F_2$ as
\begin{eqnarray} y_j=\sum_{i=1}^{n-j} x_i^j+\sum_{i=n-j+2}^n x_i^{j-1}.
\end{eqnarray}

The senders send packets to the set of receivers in rounds. Every packet
initiated from the sender $s_i$ contains $ID$, data $x_{s_i}$, and a round $t_\delta^\ell$. For example, the sender $s_i$ will send the
$packet_{s_i}$ as follows.

\begin{eqnarray}
packet_{s_i}=(ID_{s_i},x_{s_i},t_{\delta}^\ell).
\end{eqnarray}
Also, the sender $s_j$ will send the encoded data $y_{s_j}$ as

\begin{eqnarray}
packet_{s_i}=(ID_{s_j},y_{s_j},t_{\delta}^\ell).
\end{eqnarray}
 We ensure that the encoded data $y_{s_j}$ is varied per one round  transmission for every session. This means that the path $L_j$ is dedicated to send only one encoded data $y_j$ and all data
$x_j^1,x_j^2,\ldots,x_j^{n-1}$.
\begin{remark}
In NPS-I, the data transmitted from the sources does not experience any
round time delay. This means that the receivers will be able to decode the received packets online and immediately recover the failed data.
\end{remark}
\begin{lemma}
The normalized capacity NPS-I of the network model $\N$ described
in Scheme~(\ref{eq:n-1protectionII}) is given by
\begin{eqnarray}
\C=(n-1)/(n)
\end{eqnarray}
\end{lemma}
\begin{proof}
We have $n$ rounds and the total number of transmitted packets in every
round is $n$. Also, in every round there are (n-1) un-encoded data
$x_1,x_2,\ldots x_{i\neq j},\ldots,x_{n}$ and only one encoded data $y_j$, for all $i=1,\ldots,n$. Hence, the capacity $c_\ell$ in every round is
$n-1$. Therefor, the normalized capacity is given by
\begin{eqnarray}
\C=\frac{\sum_{\ell=1}^n c_\ell}{n*n}=\frac{(n-1)*n}{n^2}.
\end{eqnarray}
%See~\cite{aly08i}.
\end{proof}

 The following lemma shows that the network protection strategy
NPS-II is in fact optimal if we consider $\F_2$. In other words, there
exist no other strategies that give better normalized capacity than
NPS-II.

\begin{lemma}
The network protection scheme NPS-II against a single link failure is
optimal.
\end{lemma}

The transmission is done in rounds, hence  linear combinations of data
have to be from the same round time. This can be achieved using the round
time that is included in each packet sent by a sender.

\bigskip

\noindent \textbf{Encoding Process:} There are several scenarios where the
encoding operations can be achieved. The encoding and decoding operations
will depend mainly on the network topology; how the senders and receivers
are distributed in the network.
\begin{itemize}
\item The encoding operation is done at only one source $s_i$. In this
    case all other sources must send their data to $s_i$,  which will
    send encoded data over $L_i$. We assume that all sources have
    shared paths with each other.
\item If we assume there is a data distributor $\mathcal{S}$, then the
    source nodes send a copy of their data to the data distributor
    $\mathcal{S}$, in which it will decide which source will send the
    encoding data and all other sources will send their own data. This
    process will happen in every round during transmission time.
\item The encoding is done by the bit-wise operation which is the
    fastest arithmetic operation that can be performed among all
    source's data.
\item The distributor $\mathcal{S}$ will change the sender that should send the encoded data in every round of a given session.

\end{itemize}

\bigskip

%%%%%%%%%%%%%%%%%%%%%%%%%%%%%%%%%%%%%%%%%%%%%%%%%%%%%%%%%%%%%%%%%%%
\section{Implementation Aspects}\label{sec:implementation}

In this section we shall provide implementation aspects of the proposed protection strategy in case of a single link failure. The network protection strategy against a link failure is deployed in two processes: Encoding and decoding operations. The  encoding operations are performed at the set of sources, in which one or two sources will send the encoded data depending on the strategy used . The decoding operations are performed at the receivers' side, in which a receiver with a failed link had to Xor all other receivers' data in order to recover its own data. Depending on NPS-I or NPS-II the receivers will experience some delay before they can actually decode the packets. If the failure happen in the protection path of NPS-I, then the receivers do not perform any decoding operations because all working paths will convey data from the senders to receivers. However, if the failure happens in the working path, the receivers must perform decoding operations to recovery the failure using the protection path. We also note that the delay will happen only when the failure occurs in the protection paths.

The synchronization between senders and receivers are    done using the
time rounds, hence linear combinations of data have to be from the same
round time. Each packet sent by a sender has its own time and ID. In this
part we will assume that there is a data distributor $\mathcal{S}$ at the
sources side and data distributor $\mathcal{R}$ at the receivers side.

\bigskip

\noindent \textbf{Encoding Process:} The encoded process of the proposed
protection strategies can be done as follows.
\begin{itemize}
\item The source nodes send a copy of their data to the data
    distributor $\mathcal{S}$, then $\mathcal{S}$ will decide which
    source will send the encoding data and all other sources will send
    their own data. This process will happen in every round during
    transmission time.
\item The encoding is done by the bit-wise operation which is the
    fastest arithmetic operation that can be performed among all
    source's data.
\item The server $\mathcal{S}$ will change the sender that should send
    the encoded data in every round of a given session.
\item This process will be repeated in every session during
    transmission till all sources send their data.
\end{itemize}

\bigskip

\noindent \textbf{Decoding Process:} The decoding process is done in a
similar way as the encoding process. We assume there is a data distributor server $\mathcal{S}$ that assigns the senders that will send only their
own data as shown in Fig.~\ref{fig:impl1failure}. In addition
$\mathcal{S}$ will encode the data from all senders and distribute it only to the sender that will transmit the encoded data over its path.  The
objective is to withhold rerouting the signals or the transmitted packets
due to link failures. However, we provide strategies that utilize network
coding and reduced capacity at the source nodes.

We assume there is a data distributor $\mathcal{R}$ that will collect data from all working and protection paths and is able to perform the decoding
operations.   In this case we assume that all receivers $R$ have available shared   paths with the data collector $\mathcal{R}$. At the receivers
side, if there is  one single failure in a path $L_k$, then there are
several situations.

\begin{itemize}
\item If the path $L_k$ carries data without
encoding (it is a working path), then the data distributor
$\mathcal{R}$ must query all other nodes in order to recover the
failed data. In this case $r_k$ must query $\mathcal{R}$ to retrieve
its data.

\item If the path $L_k$ carries encoded data $y_k$, then it does not
    need to perform any action, since $y_k$ is used for protection and does not have any valued data. \end{itemize}

\begin{figure}[t]
\begin{center}
\includegraphics[width=8.6cm,height=4.2cm]{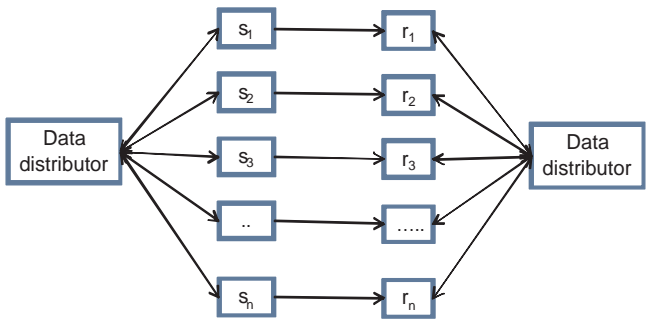}
\caption{An example that shows network protection against single failure using reduced capacity.}\label{fig:impl1failure}
\end{center}
\end{figure}

\bigskip

\section{Conclusion}\label{sec:conclusion}
In this paper we presented a strategy for network protection against a single link failure in optical networks. We showed that protecting a single link failure in optical networks can be achieved using network coding and reduced capacity. In addition, we provided implementation aspects of the proposed network protection strategies.

\bibliographystyle{ieeetr}

\end{document}